\documentclass{article}

\usepackage{fullpage}
\usepackage{amsfonts}
\usepackage{amssymb}
\usepackage{amsmath}
\usepackage{amsthm}
\usepackage{multirow}
\usepackage{booktabs}
\usepackage{xspace}
\usepackage{calc}
\usepackage{url}
\usepackage{csquotes}
\usepackage{graphicx}

\usepackage{hyperref}

\def\ve#1{\mathchoice{\mbox{\boldmath$\displaystyle\bf#1$}}
{\mbox{\boldmath$\textstyle\bf#1$}}
{\mbox{\boldmath$\scriptstyle\bf#1$}}
{\mbox{\boldmath$\scriptscriptstyle\bf#1$}}}

\newcommand\veb{{\ve b}}

\newcommand\vel{{\ve l}}

\newcommand\vep{{\ve p}}

\newcommand\ves{{\ve s}}

\newcommand\veu{{\ve u}}

\newcommand\vew{{\ve w}}
\newcommand\vex{{\ve x}}

\def \A {E^{(n)}}

\def\Z{\mathbb{Z}}
\def\N{\mathbb{N}}

\newcommand{\pref}{\ensuremath{\succ}}
\newcommand{\Oh}{O}

\DeclareMathOperator{\sgn}{{\rm sign}}
\DeclareMathOperator{\bool}{{\rm bool}}

\DeclareMathOperator{\rank}{{\rm rank}}

\newcommand{\maxcoef}{\ensuremath{\texttt{a}}}

\newcommand{\FPT}{{\sf FPT}\xspace}
\newcommand{\NP}{{\sf NP}\xspace}
\newcommand{\XP}{{\sf XP}\xspace}

\DeclareGraphicsExtensions{.pdf}

\newcommand{\defparproblem}[4]{
  \bigskip
  \noindent
  \fbox{
    \begin{minipage}{.96\linewidth}
      \textsc{#1} \hfill Parameter: #2\\[2pt]
      \smallskip
      \noindent
      \begin{tabular}{@{}l@{ }l}
        \emph{Input:} & \begin{minipage}[t]{\linewidth-\widthof{Input:\ \ }}
                          #3
                        \end{minipage}\\[2pt]
        \emph{Task:} & \begin{minipage}[t]{\linewidth-\widthof{Input:\ \ }}
                          #4
                       \end{minipage}
    \end{tabular}
  \end{minipage}
  }
  \medskip
}

\theoremstyle{plain}
\newtheorem{theorem}{Theorem}%[section]
\newtheorem{corollary}[theorem]{Corollary}
\newtheorem{observation}[theorem]{Observation}
\newtheorem{lemma}[theorem]{Lemma}
\newtheorem{proposition}[theorem]{Proposition}

\theoremstyle{definition}
\newtheorem{definition}[theorem]{Definition}

\begin{document}

% Author macros::begin %%%%%%%%%%%%%%%%%%%%%%%%%%%%%%%%%%%%%%%%%%%%%%%%
\title{Voting and Bribing in Single-Exponential Time}
\author{Du\v{s}an Knop\thanks{TU Berlin, Germany. \texttt{knop@kam.mff.cuni.cz}. Supported by project CE-ITI P202/12/G061 of GA \v{C}R and project 1784214 of GA UK.}
  \and
  Martin Kouteck{\'y}\thanks{Department of Applied Mathematics, Charles University, Prague, Czech Republic. \texttt{koutecky@kam.mff.cuni.cz}. Supported project 14-10003S of GA \v{C}R and projects 1784214 and 338216 of GA UK.}
  \and
  Matthias Mnich\thanks{Institut f{\"u}r Informatik, Universit{\"a}t Bonn, Bonn, Germany. \texttt{mmnich@uni-bonn.de}. Supported by ERC Starting Grant 306465 (BeyondWorstCase) and DFG project num. MN 59/4-1.}
}

\maketitle

\begin{abstract}
  We introduce a general problem about bribery in voting systems.
  In the $\mathcal R$-{\sc Multi-Bribery} problem, the goal is to bribe a set of voters at minimum cost such that a desired candidate wins the perturbed election under the voting rule $\mathcal R$.
  Voters assign prices for withdrawing their vote, for swapping the positions of two consecutive candidates in their preference order, and for perturbing their approval count to favour candidates.

  As our main result, we show that $\mathcal R$-{\sc Multi-Bribery} is fixed-parameter tractable parameterized by the number of candidates for many natural voting rules $\mathcal R$, including Kemeny rule, all scoring protocols, maximin rule, Bucklin rule, fallback rule, SP-AV, and any C1 rule.
  In particular, our result resolves the parameterized complexity of $\mathcal R$-{\sc Swap Bribery} for all those voting rules, thereby solving a long-standing open problem and ``Challenge \#2'' of the ``Nine Research Challenges in Computational Social Choice'' by Bredereck et al.

  Further, our algorithm runs in single-exponential time for arbitrary cost; it thus improves the earlier double-exponential time algorithm by Dorn and Schlotter that is restricted to the uniform-cost case for all scoring protocols, the maximin rule, and Bucklin rule.
\end{abstract}

\section{Introduction}\label{sec:introduction}
In this work we address algorithmic problems from the area of voting and bribing.
In these problems, we are given as input an election, which consists of a set $C$ of candidates and a set $V$ of voters $v$, each of which is equipped with a partial order $\pref_v$ indicating their preferences over (subsets of) the candidates.
Further, we have a fixed voting rule $\mathcal R$ (that is not part of the input), which determines how the orders of the voters are aggregated to determine the winner(s) of the election among the candidates.
Popular examples of voting rules $\mathcal R$ include ``scoring protocols'' like \emph{plurality}---where the candidate(s) ranked first by a majority of voters win(s)---or the Borda rule, where each candidate receives $|C|-i$ points from being ranked $i$-th by a voter and the candidate with most points wins; and the Copeland rule, which orders candidates by their number of pairwise victories minus their number of pairwise defeats.

The goal is to \emph{perturb} the given election~$(C,V)$ by \emph{bribing voters} through \emph{bribery actions}~$\Gamma$ in such a way that a designated candidate $c^\star\in C$ wins the perturbed election~$(C,V)^{\Gamma}$ under the voting rule $\mathcal R$.
Such perturbation problems model various real-life issues, such as actual bribery, campaign management, or post-election checks, as in destructive bribery (known as margin of victory); for an overview of the many flavours of bribery problems we refer to a recent survey by Faliszewski and Rothe~\cite{FaliszewskiRothe2015}.

\paragraph{$\mathcal R$-{\sc Multi-Bribery}}
Perturbation is performed by the actions of \emph{swapping} the position of two adjacent candidates in the preference order of some voter, by \emph{push actions} that perturb the approval count of a voter, and \emph{control changes} that (de)activate some voters.
The algorithmic problem is to achieve the goal by performing the most cost-efficient actions.
To measure cost of swaps, we consider the model introduced by Elkind et al.~\cite{ElkindEtAl2009} where each voter may assign different prices $\sigma^v(c,c')$ for swapping two consecutive candidates $c,c'$ in their preference order; this captures the notion of small changes and comprises the preferences of the voters.
If voter $v$ is involved in a swap or push action, a one-time influence cost $\iota^v$ occurs.
We additionally allow voter-individual cost $\pi^v$ for push actions, and voter-individual cost $\alpha^v$ and $\delta^v$ for activation and deactivation.
Our model is general enough that we even allow \emph{negative} cost.
We implicitly allow preference orders to be partial, in particular we allow top-truncated orders, where voters rank only their top candidates and are indifferent towards the other candidates (see Sect.~\ref{sec:computationalsocialchoiceproblems}).
We call this very general set-up the $\mathcal R$-{\sc Multi-Bribery} problem.

Various special cases of the $\mathcal R$-{\sc Multi-Bribery} problem have been studied in the literature; see Table~\ref{tab:results} for an overview of which problems are captured by $\mathcal R$-{\sc Multi-Bribery}:
\begin{table}[htb]
  \centering
  \resizebox{\textwidth}{!}{%
    \begin{tabular}{ll}
      \toprule
      problem name                         & specialization of $\mathcal R$-{\sc Multi-Bribery}\\
      \midrule
      $\mathcal R$-{\sc \$Bribery}         & $\sigma^v \equiv 0,\pi^v \equiv \infty, \alpha^v \equiv \delta^v \equiv \infty$\\
      \midrule
      $\mathcal R$-{\sc Manipulation}      & $\iota^v \equiv 0$ for $v\in S$, $\iota^v \equiv \infty$ for $v\notin M$\\
      \midrule
      $\mathcal R$-CCAV/$\mathcal R$-CCDV  & $\iota^v \equiv 0,\sigma^v \equiv \pi^v \equiv \infty$\\
      \midrule
      $\mathcal R$-{\sc Swap Bribery}      & $\pi^v \equiv \infty$, $\alpha^v \equiv \delta^v \equiv \infty$, $\iota^v \equiv 0$\\
      \midrule
      $\mathcal R$-{\sc Shift Bribery}     & $\mathcal R$-{\sc Swap Bribery} with $\sigma^v(c,c')=\infty$ for $c^\star\notin\{c,c'\}$\\
      \midrule
      $\mathcal R$-{\sc Support Brib.}     & $\sigma^v \equiv \alpha^v \equiv \delta^v \equiv \infty, \iota^v \equiv 0$\\
      \midrule
      $\mathcal R$-{\sc Mixed Bribery}     & $\alpha^v \equiv \delta^v \equiv \infty$, $\iota^v \equiv 0$\\
      \midrule
      $\mathcal R$-{\sc Extension Bribery} & $\sigma^v(c,c')=0$ if $c,c'$ unranked, else $\sigma^v(c,c') = \infty$; $\alpha^v \equiv \delta^v \equiv \infty$, $\iota^v = 0$\\
      \midrule
      $\mathcal R$-{\sc Possible Winner}   & reduce to \textsc{$\mathcal R$-Swap Bribery}~\cite[Thm. 2]{ElkindEtAl2009}\\
      \midrule
      {\sc Dodgson Score}                  & Condorcet-{\sc Swap Bribery} with $\sigma^v \equiv 1$\\
      \midrule
      {\sc Young Score}                    & Condorcet-CCDV with $\delta^v \equiv 1$\\
      \bottomrule
    \end{tabular}
  }
  \caption{\label{tab:specialcases}$\mathcal R$-{\sc Multi-Bribery} generalizes several studied bribery problems, whose formal definitions we give in Appendix~\ref{sec:problemdefinitions}.}
\end{table}

For instance, in the $\mathcal R$-{\sc Swap Bribery} problem, only swaps are permitted.
On the other hand, in the $\mathcal R$-{\sc Support Bribery} problem, only push actions are allowed.
And in the~{\sc Young Score} problem, only control changes are allowed.

The aforementioned many special cases of $\mathcal R$-{\sc Multi-Bribery} have been investigated extensively from an algorithmic viewpoint.
As it turns out, most of the cases looked at are $\mathsf{NP}$, such as e.g. $\mathcal R$-{\sc Swap Bribery}.
Therefore, we expect algorithms solving these problem exactly to take superpolynomial time (assuming $\mathsf{P}\neq\mathsf{NP}$).
Yet, in many application scenarios it is reasonable to assume that the number $|C|$ of candidates is small; it has therefore been of high interest to design algorithms for $\mathsf{NP}$-hard voting and bribery problem that exploit this property.
In particular, a quest for algorithms which solve instances $I$ of size $n$ in time $f(|C|)\cdot n^{O(1)}$ for some function $f$ has been made; such algorithms are called \emph{fixed-parameter algorithms}.
Fixed-parameter algorithms are contrasted with so-called $\mathsf{XP}$-algorithms which have run times of the form $n^{f(|C|)}$.
Whereas $\mathsf{XP}$-algorithms are generally considered impractical even for small instances $I$ and few candidates, fixed-parameter algorithms have the potential to be practical, provided that the function $f$ exhibits moderate growth.

The current situation for voting and bribing problems is that indeed, a large number of fixed-parameter algorithms for $\mathsf{NP}$-hard voting and bribing problems parameterized by $|C|$ have been designed, for a multitude of voting rules~\cite{BartholdiEtAl1989,BetzlerEtAl2009,BredereckEtAl2014b,BredereckEtAl2014c,BredereckEtAl2015,BredereckEtAl2016b,BredereckEtAl2016,DornSchlotter2012,FaliszewskiEtAl2011}.
For instance, a fundamental algorithm in that direction is due to Dorn and Schlotter~\cite{DornSchlotter2012}, who show how to solve $\mathcal R$-{\sc Swap Bribery} with uniform costs\footnote{Bredereck et al.~\cite{BredereckEtAl2014b} pointed out that the algorithm by Dorn and Schlotter only works for uniform costs.} in time $2^{2^{O(|C|)}}\cdot n^{O(1)}$ for all so-called linearly describable voting rules $\mathcal R$; here, uniform costs means that swapping any two candidates always has the same cost for all voters.
In the Dorn-Schlotter algorithm, as well as several other fixed-parameter algorithms designed for few candidates, the function $f(|C|)$ grows quite fast, often double-exponential in~$|C|$, which a priori makes these algorithms impractical even for very few candidates.
As the double-exponential dependence on $|C|$ in those algorithms often stems from solving, as a subroutine, a certain integer linear program (ILP) by means of a deep result of Lenstra~\cite{Lenstra1983}, Bredereck et al.~\cite{BredereckEtAl2014} put forward the following ``Challenge~\#1'', as part of their ``Nine Research Challenges in Computational Social Choice'':

\begin{displayquote}
    \emph{Many [FPT results in computational social choice] rely on a deep result from combinatorial optimization due to Lenstra [that] is mainly of theoretical interest; this may render corresponding fixed-parameter tractability results to be of classification nature only.
     %  %Fixed-parameter tractability results based on Integer Linear Programming also tend to give less insight into the structural properties of the problems than combinatorial algorithms. [...]
     Can the mentioned ILP-based [...] results be replaced by direct combinatorial [...] fixed-parameter algorithms?}
\end{displayquote}

Another downside of the ``ILP-based approach'' is that it inherently treats voters not as individuals, but as groups which share preferences.
This makes it difficult to obtain algorithms where voters from the same group differ in some way, such as by the cost of bribing them.
Their ``Challenge \#2'' thus reads:
\begin{displayquote}
      \emph{[T]here is a huge difference between [...] problems, where each voter has unit cost for being bribed, and the other flavors of bribery, where each voter has individually specified price [...] and it is not known if they are in $\mathsf{FPT}$ or hard for $\mathsf{W[1]}$.
      What is the exact parameterized complexity of the $\mathcal R$-{\sc Swap Bribery} and $\mathcal R$-{\sc Shift Bribery} parameterized by the number of candidates, for each voting rule $\mathcal R$?}
\end{displayquote}

\subsection{Our Contribution}\label{sec:ourcontribution}
Our main result is a fixed-parameter algorithm for $\mathcal R$-{\sc Multi-Bribery} parameterized by the number of candidates, for many fundamental voting rules $\mathcal R$.
Our algorithm has a few advantages over previous works, in that it works for voter-dependent cost functions, and it runs in time that is only single-exponential in $|C|$.
The algorithm applies to a large number of voting rules $\mathcal R$, such as all \emph{natural} scoring protocols where each candidate receives at most $|C|$ points from each voter.
\begin{theorem}
\label{thm:main-mixedplusbribery}
  $\mathcal R$-{\sc Multi-Bribery} is fixed-parameter tractable parameterized by the number of candidates, and can be solved in time
  \begin{itemize}
    \item $2^{O(|C|^6\log|C|)}\cdot n^{3}$ when $\mathcal R$ is any natural scoring protocol, any C1 rule,  or SP-AV,
    \item $2^{O(|C|^6\log|C|)}\cdot n^{4}$ when $\mathcal R$ is the maximin, Bucklin or fallback rule, and
     \item $2^{O(|C|!^6\log|C|)}\cdot n^{3}$ when $\mathcal R$ is the Kemeny rule.
  \end{itemize}
\end{theorem}
In summary, our algorithm subsumes, and improves, \emph{all} previously devised algorithms for the problems listed in Table~\ref{tab:specialcases}.
For some problems, such as {\sc Young Sore}, our algorithm yields the first improvement since 1977; we summarize the comparison in Table~\ref{tab:results}.
\begin{table}[htb]
  \centering
  \resizebox{\textwidth}{!}{%
     \begin{tabular}{llll}
       \toprule
                                             & \multicolumn{2}{c}{previous best result}                                                           & new result\\
       problem                               & run time / hardness                                & voting rules $\mathcal R$                     &\\
       \midrule
       $\mathcal R$-{\sc \$Bribery}          & $2^{2^{O(|C|)}} n^{O(1)}$                          & Approval~\cite{BredereckEtAl2015}             & $2^{O(|C|^6\log|C|)}\cdot n^{3}$ \\
       \midrule
       $\mathcal R$-{\sc Manipulation}       & $2^{2^{O(|C|)}} n^{O(1)}$                          & Borda~\cite{BetzlerEtAl2011}                  & $2^{O(|C|^6\log|C|)}\cdot n^{3}$ \\
       \midrule
       $\mathcal R$-CCAV/$\mathcal R$-CCDV   & $2^{2^{O(|C|)}} n^{O(1)}$                          & Approval~\cite{BredereckEtAl2015}             & $2^{O(|C|^6\log|C|)}\cdot n^{3}$ \\
       \midrule
       $\mathcal R$-{\sc Swap Bribery}       & $2^{2^{O(|C|)}} n^{O(1)}$, uniform cost            & Approval~\cite{DornSchlotter2012}             & $2^{O(|C|^6\log|C|)}\cdot n^{3}$ \\
       \midrule
       $\mathcal R$-{\sc Shift Bribery}      & \XP-algor., arbitrary cost,                     & Borda, Cope-\\
                                             & \FPT-AS, restricted cost                           & land, Maximin~\cite{BredereckEtAl2016b}       & $2^{O(|C|^6\log|C|)}\cdot n^{4}$ \\
       \midrule
       $\mathcal R$-{\sc Support Brib.}      & \NP-hard                                           & Fallback, SP-AV~\cite{SchlotterEtAl2017}      & $2^{O(|C|^6\log|C|)}\cdot n^{4}$ \\
       \midrule
       $\mathcal R$-{\sc Mixed Bribery}      & \NP-hard                                           & SP-AV~\cite{ElkindEtAl2009}                   & $2^{O(|C|^6\log|C|)}\cdot n^{3}$ \\
       \midrule
       $\mathcal R$-{\sc Extension Bribery}  & \NP-hard                                           & Borda, Cope-\\
                                             &                                                    & land$^{0}$, Maximin~\cite{BaumeisterEtAl2012} & $2^{O(|C|^6\log|C|)}\cdot n^{4}$ \\
       \midrule
       $\mathcal R$-{\sc Possible Winner}    & $2^{2^{O(|C|)}} n^{O(1)}$                          & Bucklin, Copeland,                             &\\
                                             &                                                    & pos. scoring protocols~\cite{BetzlerEtAl2009} & $2^{O(|C|^6\log|C|)}\cdot n^{3}$ \\
      \midrule
      {\sc Dodgson Score}                    & $2^{2^{O(|C|)}} n^{O(1)}$~\cite{BartholdiEtAl1989} & & $2^{O(|C|^6\log|C|)}\cdot n^{3}$ \\
      \midrule
      {\sc Young Score}                      & $2^{2^{O(|C|)}} n^{O(1)}$~\cite{Young1977}         & & $2^{O(|C|^6\log|C|)}\cdot n^{3}$ \\
      \bottomrule
    \end{tabular}}
    \caption{\label{tab:results}Summary of results from this paper for $\mathcal R$-{\sc Multi-Bribery} compared to previous works for special cases.
      In each row corresponding to a problem $\mathcal R$-{\sc Problem}, we state the previous best known dependency on $|C|$ for voting rules $\mathcal R$ for which $\mathcal R$-{\sc Problem} is known to be $\mathsf{NP}$-hard.
      For $\mathcal R$ being the Kemeny rule, no previous results are known to us.
      For $\mathcal R$-{\sc Shift Bribery}, \FPT-AS refers to a \emph{fixed-parameter approximation scheme} which is an algorithm that yields a $(1-\varepsilon)$-approximate solution in time $f(1/\varepsilon,|C|)\cdot n^{O(1)}$ for some superpolynomial function $f$ and any $\varepsilon > 0$; it is thus a weaker result than a fixed-parameter algorithm.
    }
\end{table}

\paragraph{Applications of Theorem~\ref{thm:main-mixedplusbribery}.}
We argued that $\mathcal R$-{\sc Multi-Bribery} generalizes many well-studied voting and bribing problems, parameterized by the number of candidates.
A direct corollary of Theorem~\ref{thm:main-mixedplusbribery} is:
\begin{corollary}
    \label{thm:swapbribery}
    Let $\mathcal R$ be any natural scoring protocol, a C1 rule, the maximin rule, the Bucklin rule, the SP-AV rule, the fallback rule, or Kemeny rule.
    Then $\mathcal R$-{\sc Swap Bribery} is fixed-parameter tractable parameterized by the number $|C|$ of candidates.
\end{corollary}
This solves ``Challenge \#2'' by Bredereck et al.~\cite{BredereckEtAl2014}.
In particular, for scoring protocols, maximin rule and Bucklin rule, Corollary~\ref{thm:swapbribery} extends and improves an algorithm by Dorn and Schlotter~\cite{DornSchlotter2012} that is restricted to the uniform cost case of $\mathcal R$-{\sc Swap Bribery}, and requires double-exponential run time $2^{2^{O(|C|)}}\cdot n^{O(1)}$.
% achieve the exponential improvements over previous run times for $\mathcal R$-{\sc Swap Bribery}.

We remark that it is unclear (cf. \cite[p. 338]{FaliszewskiEtAl2011}) if the Kemeny rule can be \emph{described by linear inequalities} as defined by Dorn and Schlotter~\cite{DornSchlotter2012}; even if it does, ours is the first fixed-parameter algorithm for $\mathcal R$-{\sc Swap Bribery} under the Kemeny rule, as Dorn and Schlotter's algorithm only applies to the unit-cost case.

Another corollary of Theorem~\ref{thm:main-mixedplusbribery} is the following:
\begin{corollary}
  $\mathcal R$-{\sc Shift Bribery} is fixed-parameter tractable parameterized by the number of candidates, for $\mathcal R$ being the Borda rule, the maximin rule and the Copeland$^\alpha$ rule.
\end{corollary}
This way, we simultaneously improve the fixed-parameter algorithm by Dorn and Schlotter~\cite{DornSchlotter2012} for uniform cost, the $\mathsf{XP}$-algorithm and the fixed-parameter approximation scheme for arbitrary cost by Bredereck et al.~\cite{BredereckEtAl2014b}.

Further, we have the following:
\begin{corollary}
  Approval-{\sc \$Bribery}, Approval-\$CCAV and Approval-\$CCDV can be solved in time\linebreak $2^{O(|C|^6\log|C|)}\cdot n^{4}$.
\end{corollary}
This improves a recent result by Bredereck et al.~\cite{BredereckEtAl2015} who solved these problems in time that is double-exponential in $|C|$.

\subsection{Our Approach}
The run times that we achieve in Theorem~\ref{thm:main-mixedplusbribery} are (except for the Kemeny rule) only single-exponential in $|C|$.
To achieve this, we avoid using Lenstra's algorithm for solving fixed-dimension ILPs~\cite{Lenstra1983}, which was the method of choice so far (and which led to double-exponential run times).
Typically, when using Lenstra's algorithm one has to ``group objects'' in order to be able to bound their number in terms of the used parameters.
Instead, we formulate the $\mathcal R$-{\sc Multi-Bribery} problem in terms of an $n$-fold integer program (IP).
Unlike fixed-dimension ILPs, $n$-fold IPs allow variable dimension at the expense of a more rigid block structure of the constraint matrix.
We manage to encode the $\mathcal R$-{\sc Multi-Bribery} problem for many voting rules $\mathcal R$ in a constraint matrix that has this required structure.
The formulations are not straightforward: rather, we model the problems in terms of an ``extended'' $n$-fold IP which has a more general format than is required by the ``standard'' $n$-fold IP discussed in the literature.
Then we show how to efficiently transform any extended $n$-fold IP into a standard $n$-fold IP.
The dimension of this IP is not bounded by \emph{any function} of the number $|C|$ of candidates; however, we bound the dimension of each block by a \emph{polynomial} in $|C|$.
Then we solve the standard $n$-fold IP via an algorithm of Hemmecke et al.~\cite{HemmeckeEtAl2013} whose runtime depends exponentially only on the largest dimension of each of its blocks.
That algorithm has a rather combinatorial flavour by traversing the problem's ``Graver bases''; cf. Hemmecke et al.~\cite{HemmeckeEtAl2013}.
This way, we substantially contribute towards resolving ``Challenge \#1'' by Bredereck et al.~\cite{BredereckEtAl2014}.

\subsection{Related Work}\label{sec:relatedwork}
Bribery problems in voting systems are well-studied~\cite{BredereckEtAl2014b,DornSchlotter2012,ElkindEtAl2009,FaliszewskiEtAl2011}.
Bredereck et al.~\cite{BredereckEtAl2014b} consider shift bribery, where candidates can be shifted up a number of positions in a voter's preference order; this is a special case of swap bribery.
An extension of their model~\cite{BredereckEtAl2016} allows campaign managers to affect the position of the preferred candidate in multiple votes, either positively or negatively, with a single bribery action, which applies to large-scale campaigns.
In a different model~\cite{BredereckEtAl2016b}, complexity of bribery of elections admitting for multiple winners, such as when committees are formed, has been studied.
Also, different cost models have been considered: Faliszewski et al.~\cite{FaliszewskiEtAl2009} require that each voter has their own price that is independent of the changes made to the bribed vote.
The more general models of Faliszewski~\cite{Faliszewski2008} and Faliszewski et al.~\cite{FaliszewskiEtAl2009b} allow for prices that depend on the amount of change the voter is asked for by the briber.
For various other bribery models that have been investigated algorithmically, cf. Rothe~\cite[Chapter 4.3.5]{BaumeisterRothe2016}.

Regarding ILPs, tractable fragments include ILPs whose defining matrix is totally unimodular (due to the integrality of the corresponding polyhedra and the polynomiality of linear programming), and ILPs in fixed dimension (due to the algorithms of Lenstra~\cite{Lenstra1983} and Kannan~\cite{Kannan1983}).
Courcelle's theorem~\cite{Courcelle1990} implies that solving ILPs is fixed-parameter tractable parameterized by the treewidth of the constraint matrix and the maximum domain size of the variables.
Ganian and Ordyniak~\cite{GanianOrdyniak2016} showed fixed-parameter tractability for the combined parameter the treedepth and the largest absolute value in the constraint matrix, and contrasted this with a $\mathsf{W}[1]$-hardness result when treedepth is exchanged for treewidth.
Recently, further combinatorial algorithms for special cases of $n$-fold IPs and related tree-fold IPs have been suggested~\cite{ChenMarx2018,KnopEtAl2017b}.

\paragraph{Organization.}
In Sect.~\ref{sec:computationalsocialchoiceproblems} we provide the necessary background on the problems that we solve.
Then in Sect.~\ref{sec:extendednfoldips} we define extended $n$-fold IPs, which later allows for easier problem modeling.
We do so in Sect.~\ref{sec:applications}, where we give extended $n$-fold IP formulations for several instantiations of the $\mathcal R$-{\sc Multi-Bribery} problem.
The complexity lower bounds and hardness results are given in Sect.~\ref{sec:lowerbounds}.
We conclude in Sect.~\ref{sec:conclusions}.

\section{Voting and Bribing Problems}
\label{sec:computationalsocialchoiceproblems}
We give notions for the problems we deal with; for background, we refer to the surveys of Brams and Fishburn~\cite{BramsFishburn2002}, and Faliszewski and Rothe~\cite{FaliszewskiRothe2015}.

\medskip
\noindent
\textbf{Elections.}
An election~$(C,V)$ consists of a set $C$ of candidates and a set~$V$ of voters, who indicate their preferences over the candidates in $C$.
There are many ways in which a voter's preferences can be modeled; throughout this paper we use a variant of the ordinal model, where each voter $v$'s preferences are represented via a \emph{preference order} $\pref_v$ which is a partial order over $C$.
We explicitly allow partial orders; we call an election \emph{complete} if all voters' preference orders are linear orders.
In some problems we study voters who indicate their preferences only for their ``top candidates''; we model this with ``truncated orders''.
For an integer $t\in\mathbb N$, a preference order~$\pref_v$ is \emph{$t$-top-truncated} if there is a permutation $\pi$ over $\{1,\hdots,|C|\}$ such that $\pref_v$ is of the form $c_{\pi(1)} \pref_v \cdots \pref_v c_{\pi(t)} \pref_v \left\{c_{\pi(t+1)},\hdots,c_{\pi(|C|)}\right\}$; that is, $v$ is indifferent among the members of the set $\left\{c_{\pi(t+1)},\hdots,c_{\pi(|C|)}\right\}$ which we call \emph{unranked candidates}; we refer to $\left\{c_{\pi(1)},\hdots,c_{\pi(t)}\right\}$ as to the \emph{ranked candidates}.
For a ranked candidate~$c$ we denote by $\textrm{rank}(c,v)$ their rank in~$\pref_v$; then $v$'s most preferred candidate has rank~1 and their least preferred candidate has rank $|C|$.
Also, for $t$-top-truncated preference orders~$\pref_v$ it holds that $\textrm{rank}(c,v) \leq t$ for all \emph{ranked} candidates $c\in C$.
If $\pref_v$ is not a linear order (thus, $\pref_v$ is a quasiorder), we replace it with any linear extension of $\pref_v$ and set the cost of swapping two candidates $c,c'$ to $0$ whenever $\textrm{rank}(c,v) = \textrm{rank}(c',v)$ in the original order.
Independently of voters, for the set of candidates $C$ we also refer to a linear order~$\pref_{C}$ over $C$ as to a \emph{ranking} of $C$ (i.e., ranking is a shorthand for a linear order on candidates).
For distinct candidates~$c,c'\in C$, we write $c\pref_v c'$ if voter~$v$ prefers~$c$ over~$c'$.
To simplify notation, we sometimes identify the candidate set~$C$ with the set $\{1,\hdots,|C|\}$, in particular when expressing permutations over $C$.
All studied problems designate a candidate in $C$; we always denote it by~$c^\star$.
We sometimes identify a voter $v$ with their preference order $\pref_v$, as long as no confusion arises.

Next, we describe the actions by which we perturb a given election $(C,V)$.
Applying a set $\Gamma$ of actions to $(C,V)$ yields a \emph{perturbed} election that we denote by $(C,V)^{\Gamma}$.
Performing an action incurs a cost; we specify these costs by functions that for each voter specify their individual cost of performing the action.

\medskip
\noindent
\textbf{Swaps.}
Let $(C,V)$ be an election and let $\pref_v\in V$ be a voter.
For candidates $c,c'\in C$, a \emph{swap} $s = (c,c')_v$ means to exchange the positions of $c$ and $c'$ in $\pref_v$; denote the perturbed order by $\pref_v^s$.
A swap~$(c,c')_v$ is \emph{admissible in $\pref_v$} if $\rank(c,v) = \rank(c',v)-1$.
A set $S$ of swaps is \emph{admissible in $\pref_v$} if they can be applied sequentially in~$\pref_v$, one after the other, in some order, such that each one of them is admissible.
Note that the perturbed vote, denoted by $\pref_v^S$, is independent from the order in which the swaps of $S$ are applied.
We also extend this notation for applying swaps in several votes and denote it $V^S$.
We specify $v$'s cost of swaps by a function $\sigma^v \colon C\times C\rightarrow \mathbb{Z}$.
A special case of swaps are ``shifts'', where we want to make $c^\star$ win the perturbed election by shifting them forward in some of the votes, at an appropriate cost, without exceeding a given budget.
Shifts can be modelled by swaps only involving $c^\star$.

\medskip
\noindent
\textbf{Push actions.}
Let $(C,V)$ be an election.
In certain voting rules, such as SP-AV or Fallback, each voter $v\in V$ additionally has an \textit{approval count} $a^v \in \{0,\hdots,|C|\}$.
Voter $v$'s approval count\footnote{See Scoring protocols for the definition.} $a^v$ indicates that they approve the top-ranked $a^v$ many candidates in their preferences order, and disapprove all others.
%notion of push action is due to Schlotter et al.~\cite{SchlotterEtAl2017}.
A ``push action'' can change a voter's approval count: formally, for voter $\pref_v\in V$ and $t \in \{-a^v,\ldots,|C|-a^v\}$, a \emph{push action}~$p^v = t$ means to change their approval count to $a^v+t$.
We specify the cost of push actions by a function $\pi^v \colon \left\{ -a^v,\ldots,|C|-a^v \right\} \to \mathbb{Z}$; we stipulate that $\pi^v(0)=0$.

If voter $v$ is involved in a swap or push action, a one-time \textit{influence cost} $\iota^v$ occurs.

\medskip
\noindent
\textbf{Control changes.}
Let $(C,V)$ be an election.
We partition the set $V$ into a set~$V_a$ of \emph{active} voters and a set~$V_\ell$ of \emph{latent} voters.
Only active voters participate in an election, but through a ``control change'' latent voters can become active or active voters can become latent.
(If no partition of $V$ into $V_a$ and $V_{\ell}$ is specified, then we implicitly assume that $V = V_a$.)

Formally, a \emph{control change} $\gamma$ activates some latent voters from~$V_\ell$ and deactivates some active voters from~$V_a$; denote the changed set of voters by $(V_\ell \cup V_a)^\gamma$.
We denote the cost of activating voter $v \in V_\ell$ by~$\alpha^v$ and the cost of deactivating voter $v \in V_a$ by $\delta^v$.

\medskip
\noindent
\textbf{Voting rules.}
A voting rule~$\mathcal R$ is a function that maps an election $(C,V)$ to a subset $W\subseteq C$, called the \emph{winners}.
We study the following voting rules:

\textit{Scoring protocols.}
A scoring protocol is defined through a vector $\ves = (s_1,\hdots,s_{|C|})$ of integers with \mbox{$s_1\geq \hdots\geq s_{|C|} \geq 0$}.
For each position $p\in\{1,\hdots,|C|\}$, value $s_p$ specifies the number of points that each candidate $c$ receives from each voter that ranks $c$ as $j^{\text{th}}$ best.
Any candidate with the maximum number of points is a winner.
Examples of scoring protocols include the Plurality rule with $\ves = (1,0,\ldots,0)$, the $d$-Approval rule with $\ves = (1,\ldots,1,0,\ldots,0)$ with $d$ ones, and the Borda rule with $\ves = (|C|-1, |C|-2, \ldots, 1, 0)$.
Throughout, we consider only \emph{natural} scoring protocols for which $s_1 \leq |C|$; this is the case for the aforementioned popular rules.

\textit{Bucklin.}
The \textit{Bucklin winning round} is the (unique) number $k$ such that using the $k$-approval rule yields a candidate with more than $\frac{n}{2}$ points, but the $(k-1)$-approval rule does not.
A \textit{Bucklin winner} is then any candidate with the maximum points (over all candidates) when the $k$-approval rule is applied.

\textit{Condorcet-consistent rules.}
A candidate~$c\in C$ is a \emph{Condorcet winner} if any other~$c'\in C \setminus \{c\}$ satisfies $|\{\pref_v \in V \mid c\pref_v c' \}| > |\{v \in V \mid c' \pref_v c\}|$.
A voting rule is \emph{Condorcet-consistent} if it selects the Condorcet winner in case there is one.
Fishburn~\cite{Fishburn1977} classified all Condorcet-consistent rules as C1, C2 or C3, depending on the kind of information needed to determine the winner.
For candidates $c,c' \in C$ let $v(c,c')$ be the number of voters who prefer $c$ over $c'$, that is, $v(c,c') = |\{\pref_v \in V \mid c \pref_v c'\}|$;
we write $c <_M c'$ if $c$ beats $c'$ in a head-to-head contest, that is, if $v(c,c') > v(c',c)$.
\begin{description}
    \item[C1:] \emph{$\mathcal R$ is C1} if knowing $<_M$ suffices to determine the winner, that is, for each pair of candidates $c,c'$ we know whether $v(c,c') > v(c',c), v(c,c') < v(c',c)$ or $v(c,c') = v(c',c)$.
    An example is the Copeland$^\alpha$ rule for $\alpha \in [0,1]$, which specifies that for each head-to-head contest between two distinct candidates, if some candidate is preferred by a majority of voters then they obtain one point and the other candidate obtains zero points, and if a tie occurs then both candidates obtain $\alpha$ points; the candidate with largest sum of points wins.
    %also: Slater rule
    \item[C2:] \emph{$\mathcal R$ is C2} if it is not C1 and knowing the \emph{exact value of} $v(c,c')$ for all $c,c'\in C$ suffices to determine the winner.
    Examples are the \emph{Maximin} rule which declares any candidate $c \in C$ a winner who maximizes $v_*(c) = \min\{v(c,c') \mid c' \in C\setminus\{c\}\}$; and the \emph{Kemeny} rule which declares any candidate $c \in C$ a winner for whom there exists a ranking $\pref_R$ of $C$ that ranks~$c$ first and maximizes the total agreement with voters
    \[
    \sum_{v\in V} \bigl| \left\{(c,c') \mid ((c\pref_R c') \Leftrightarrow (c \pref_v c')) \enspace \forall c,c' \in C \right\} \bigr|
    \]
    among all rankings.
    \item[C3:] \emph{$\mathcal R$ is C3} if it is neither C1 nor C2.
    Examples are the \emph{Dodgson} rule which declares any candidate $c\in C$ a winner for whom a minimum number of swaps make them the Condorcet winner of the manipulated election; and the \emph{Young} rule which declares any candidate $c\in C$ a winner for whom removing a minimum number of voters from the election makes $c$ the Condorcet winner of the perturbed election.
\end{description}

Additionally, if approval counts are given for each voter, other voting rules are possible:

\textit{SP-AV.} Each candidate $c$ receives a point from every voter $v$ with $\rank(c,v) \leq a^v$.
A candidate with maximum number of points wins the election.

\textit{Fallback.} Delete, for each voter $v$, their unranked candidates (i.e., all $c$ with $\rank(c,v) > a^v$) from its order.
Then, use the Bucklin rule, which might fail to determine a winner due to the deletion of unranked candidates; in that case, use the SP-AV rule.

At this point we can formally define the $\mathcal R$-{\sc Multi-Bribery} problem:\hfill\vspace{1.2mm}
\defparproblem{$\mathcal R$-Multi Bribery}
{$|C|$}
{A complete election $(C,V)$ with active voters $V_a$, latent voters $V_{\ell}$ and approval counts $a^v$ for $v\in V$, a designated candidate $c^\star\in C$, and swap costs $\sigma^v$ for $v\in V$, push action costs $\pi^v$ for $v\in V$, activation costs $\alpha^v$ for $v\in V_{\ell}$ and deactivation costs $\delta^v$ for $v\in V_a$, and a one-time influence cost $\iota^v$.\hfill\vspace{1.2mm}}
{Find a set $S$ of admissible swaps, a set $P$ of push actions, and a control change $\gamma$ of minimum cost so that~$c^\star$ wins the election $\left(C, \left(\left(\left( V_a\cup V_\ell \right)^\gamma \right)^{S}\right)^P\right)$ under rule $\mathcal R$.}

\section{Extended $n$-fold Integer Programming}
\label{sec:extendednfoldips}
In this section we discuss the class of $n$-fold IPs, and show how to enhance them to obtain ``extended'' $n$-fold IPs.

\subsection{\texorpdfstring{$\boldsymbol{n}$-fold}{n-fold} integer programs}
We begin by defining $n$-fold IPs.
For background, we refer to the books of Onn~\cite{Onn2010} and De Loera et al.~\cite{DeLoeraEtAl2013}.

Let $r,s,t,n$ be positive integers.
Given $nt$-dimensional integer vectors $\vew, \veb, \vel, \veu$, an $n$-fold IP problem $(IP)_{E^{(n)},\vew,\veb,\vel,\veu}$ in variable dimension~$nt$ is defined as
\begin{equation}
\min\left\{\vew\vex\,\mid\,\A\vex=\veb\,,\ \vel\leq\vex\leq\veu\,,\ \vex\in\Z^{nt}\right\},
\label{eq:standard_nfold}
\end{equation}
where
\begin{equation*}
%\quad \mbox{where }~E^{(n)}:=
E^{(n)}:= \left(
\begin{array}{cccc}
D      & D      & \cdots & D    \\
A      & 0      & \cdots & 0      \\
0      & A      & \cdots & 0      \\
\vdots & \vdots & \ddots & \vdots \\
0      & 0      & \cdots & A    \\
\end{array}
\right)
\end{equation*}
is an $(r+ns)\times nt$-matrix, $D \in \Z^{r \times t}$ is an $r\times t$-matrix and $A \in \Z^{s \times t}$ is an $s\times t$-matrix.
For numbers, vectors and matrices, we denote by $\langle \bullet \rangle$ the binary encoding length of an object.

Hemmecke et al.~\cite{HemmeckeEtAl2013} developed a dynamic program to show the following:
\begin{proposition}[{\cite[Thm. 6.2]{HemmeckeEtAl2013}}]
    \label{thm:nfold}
    There is an algorithm that, given $(IP)_{E^{(n)},\vew,\veb,\vel,\veu}$, in time $\maxcoef^{O(trs + t^2s)}\cdot O(n^3 \langle \vew, \veb, \vel, \veu \rangle)$ either
    \begin{enumerate}
    \item declares the program infeasible or unbounded or
    \item finds a minimizer of it;
    \end{enumerate}
    where $\maxcoef = \max\{\|D\|_\infty, \|A\|_\infty\}$.
\end{proposition}
%Thus, the problem is fixed-parameter tractable parameterized by the dimensions of $A$ and $D$ and~$a$.
%Also, minimizing $p$-piecewise linear separable convex functions over $(IP)_{E^{n},\veb,\vel,\veu}$ is in \FPT.

The structure of $E^{(n)}$ allows us to divide the $nt$ variables into $n$ \textit{bricks} of size~$t$.
We use subscripts to index within a brick and superscripts to denote the index of the brick, i.e., $x_j^i$ is the $j$-th variable of the $i$-th brick with $j \in \{1, \ldots, t\}$ and $i \in \{1, \ldots, n\}$.

\subsection{Extended $n$-fold integer programs}
We now introduce a class of IPs that we call \emph{extended $n$-fold IPs}.
Our motivation for this is to enhance $n$-fold IPs with ``integer programming tricks'' that are well-known for general IPs; we want to show how to implement them while preserving the structure of $n$-fold IPs.
These tricks will make the application of $n$-fold IPs more convenient; they include introducing inequalities using slack variables, implementing logical connectives or the bool($\cdot$) operation.

That leads us to the following definition.
\begin{definition}
    Let $\vex = \left(x_1^1, \dots, x_t^1, \dots, x_1^n, \dots, x_t^n\right)$ be an $nt$-dimensional vector of integer variables.
    Let $B \in \Z$, $(b^1, \dots, b^n) \in \Z^n$ and $a=(a_1, \dots, a_t) \in \Z^t$. We say that
    \begin{align*}
    \sum_{i=1}^n \sum_{j=1}^t a_j x_j^i &= B
    \end{align*}
    is a \emph{globally uniform constraint}, and that
    \begin{align*}
    \sum_{j=1}^t a_j x_j^i &= b^i, \quad i=1, \dots, n
    \end{align*}
    is a \emph{locally uniform constraint}.
    We call the respective left hand sides \emph{globally uniform expressions} and \emph{locally uniform expressions}.
    We stress that the coefficients $a_j$ are the same regardless of the index $i$.
\end{definition}
Observe that every $n$-fold IP consists of box constraints (i.e., lower and upper bounds represented by vectors $\vel$ and $\veu$), a weight vector~$\vew$, and collections of globally and locally uniform constraints.
From now on we call the problem \eqref{eq:standard_nfold} an \emph{$n$-fold IP in standard form}.
We say that expression is a \emph{uniform expression} if it is either a globally uniform expression or a locally uniform expression.
We say that a uniform expression is a \emph{binary expression} if it is guaranteed to take only values in $\{ 0,1 \}$.

\begin{definition}
    An \emph{extended $n$-fold IP} is a collection of locally and globally uniform constraints which are additionally allowed to contain
    \begin{itemize}
        \item inequalities $<, \leq, >, \geq$,
        \item negation ($\neg$) and logical dijcunction $\lor$ with standard interpretation if applied to binary expressions and $\mbox{undefined}$ otherwise.
    \end{itemize}
    In addition to that we introduce two collections of operations $\bool_m(\cdot), \sgn_m(\cdot)$ for every positive integer $m$:
        \[
        \bool_m(x) = \begin{cases}
        0,~&\mbox{if}~x = 0,\\
        1,~&\mbox{if}~x \neq 0 \mbox{ and }-m \leq x \leq m,\\
        \mbox{undefined},~&\mbox{otherwise};
        \end{cases}
        \]
        \[
        \sgn_m(x) = \begin{cases}
        0,~&\mbox{if}~x = 0,\\
        1,~&\mbox{if}~1 \leq x \leq m,\\
        -1,~&\mbox{if}~-m \leq x \leq -1,\\
        \mbox{undefined},~&\mbox{otherwise} \enspace .
        \end{cases}
        \]
    Note that the constraints are still required to be uniform.
    An extended $n$-fold IP is \emph{valid} if the result of all of its constraints are defined (i.e. no expression used during the construction results in $\mbox{undefined}$).
    The two important parameters of a valid extended $n$-fold IP are
    \begin{description}
    \item[\emph{extended width}] is the number of inequalities, logical operations, $\bool_m(\cdot)$, and $\sgn_m(\cdot)$ operations (counting them only once for all $n$ locally uniform constraints or expressions); and
    \item[\emph{height}] is the maximum $m$ occurring in any of its $\bool_m$ and $\sgn_m$ operations.
    \end{description}
\end{definition}

The notion of height applies naturally also to expressions and constraints: the \emph{height of an expression} is the maximum $m$ appearing in any $\bool_m$ or $\sgn_m$ operation contained in it, and the \emph{height of a constraint} is the height of the expression on the left hand side.
Observe that if $\bool_m(e)$ is defined for a given expression $e$ (i.e., $-m\le e\le m$), then $\bool_{m+1}$ is defined.
Thus, we may actually require all parameters $m$ in the above definition to be the same value.

\begin{observation}
Locally uniform constraints can be of the form $x \neq_m y$.
\end{observation}
\begin{proof}
The expression $x \neq_m y$ is in fact a shorthand for $\bool_m(x-y)=1$.
\end{proof}
\begin{observation}[Folklore; see e.g.~\cite{Enderton1972}]
It is possible to use all logical connectives (i.e., ${\{\land, \ldots \}}$) with binary expressions in an extended $n$-fold IP.
\end{observation}

The following theorem deals with how many auxiliary variables and constraints are needed to convert an extended $n$-fold IP into standard form.
\begin{theorem}
    \label{thm:extendednfold}
    Let $I$ be a valid extended $n$-fold IP with  $nt$ variables, $r$ globally uniform constraints, $s$ locally uniform constraints, largest absolute coefficient value $\maxcoef$, extended width~$w$ and height $M$.
    There is an algorithm that, given $I$, in time $\Oh(ntw(r+s))$ constructs a standard $n$-fold IP $I'$ with $nt'$ variables, $r'$ globally uniform constraints, $s'$ locally uniform constraints and largest absolute coefficient value $\maxcoef'$ such that
    \begin{itemize}
        \item $t' = t + \Oh(w)$,
        \item $r' = r$,
        \item $s' = s + \Oh(w)$, and,
        \item $\maxcoef' = \max(\maxcoef, M)$.
    \end{itemize}
    Thus, $I$ can be solved in time $(\maxcoef')^\omega \cdot n^3L$ with $\omega=\Oh\big((t+w)(s+w)r + (t+w)^2(s+w)\big)$ and $L = \langle \vew, \veb, \vel, \veu \rangle$.
\end{theorem}

\paragraph{An example of an extended $n$-fold IP}
We want $m$ variables $x_1, \ldots, x_m$ in each brick describing a permutation of $\{1,\hdots,m\}$.
That is equivalent to imposing $\sum_{j=1}^m x_j = \binom{m+1}{2}$, and $x_j \neq x_k$ for all $j \neq k$, and $x_j \in \{1,\hdots,m\}$ for all $j$.
Note that in order to use a non-equality constraint ``$e \neq f$'' we need to determine the largest possible difference $|e|-|f|$.
In this case it is $m$, and thus we express the aforementioned conditions by the following $1 + \binom{m}{2}$ locally uniform constraints:
\begin{align}
\label{eqn:permut_start}
\sum_{j=1}^m x_j &= \binom{m+1}{2} \\
x_j &\neq_m x_k &\mbox{ for all } j \not = k, j \in [m] \label{eqn:permut_neq}.
\end{align}

Further, given $n$ permutations $o^i_1, \ldots, o^i_m$ for $i = 1,\hdots,n$, one for each brick.
Now we want to compare the $i$-th brick permutation $x^i_1, \ldots x^i_m$ to $o^i_1, \ldots, o^i_m$ and determine which indices are inverted, that is, $x_j < x_k$ if and only if $o_j > o_k$.
In other words, we want to determine when the sign of $(x^i_j - x^i_k)$ equals the sign of $(o^i_k - o^i_j)$.
So, for each $j \neq k$ we add a new binary indicator variable $s_{jk}$ as follows:
\begin{equation}
\label{eqn:inverted_pairs}
s_{jk} = \bool_2\left(\sgn_m\left( x^i_j - x^i_k \right) = \sgn_m\left( o^i_k - o^i_j \right)\right) \,.
\end{equation}
Notice that $\sgn_m(o^i_k - o^i_j)$ is a constant $o_{kj}^i$, so \eqref{eqn:inverted_pairs} turns to $\bool_2\left(\sgn_m(x^i_j - x^i_k) - o^i_{kj}\right)$; this is now clearly a locally uniform expression.

\iffalse
An example of what is \textit{not} expressible in the extended $n$-fold IP format is the (nonsensical) expression $\sum_{j=1}^m o^i_j x_j$.
This is not a uniform expression since the numbers $o^i_j$ appear not as additive constants but as coefficients.
However, coefficients are required to be identical across bricks.\dkcom{I do not understand this}
\fi

\subsection{Proof of Theorem~\ref{thm:extendednfold}}
\iffalse
Before we proceed to the proof we introduce some more notation for extended $n$-fold IPs.
\dkcom{I am actually not sure whether we need all of this}
We say that an operation $o$ \emph{contains} an expression $e$ if it is of form $o(e)$.
A constraint \emph{contains} an operation if it is used in it.
A constraint $c$ \emph{contains} an expression if it is used directly in $c$ (that is, not only used in operations used in $c$).
Let $I$ be a valid extended $n$-fold IP instance.
The \emph{dependency graph} of $I$ is a directed graph whose vertices are the expressions, operations, and constraint of $I$ and there is an arc $ef$ (from $e$ to $f$) if $e$ is contained in $f$.
Note that the dependency graph of $I$ is a directed acyclic graph (DAG).
It is worth noting that the number of vertices of the dependency graph is exactly the extended height of $I$.
\fi

\begin{proof}[Proof of Theorem~\ref{thm:extendednfold}]
We will prove Theorem~\ref{thm:extendednfold} by exhaustively applying a set of rewriting rules.
These rewriting rules are applied to expressions or constraints containing logical operations, the $\bool_m(\cdot)$ or $\sgn_m(\cdot)$ operations, or inequalities.
For simplicity, we will use the letters $e,f$ and $g$ for expressions as if they were single variables, and analogously $e \heartsuit f$ with $\heartsuit \in \{<, \leq, >, \geq\}$ for constraints.
This is without loss of generality, as we always create an auxiliary variable $x_e$ and assign an expression $e$ to it by adding an auxiliary constraint $x_e=e$; similarly for constraints.
A rewriting rule can be applied anytime it operands are already represented by such an variable.

To determine the parameters $r',s'$, and $t'$ of the resulting $n$-fold IP in standard form, we consider the ``$s$-increase'' $\Delta s(e)$ of an expression $e$, which is the number of auxiliary equalities required to rewrite $e$ into the standard form.
Similarly, the ``$t$-increase'' $\Delta t(e)$ of $e$ is the number of auxiliary variables needed to express~$e$, and analogously for globally uniform constraints and $\Delta r(e)$.
We note the $s$- and $t$-increase of each rule after defining it.

\medskip
\noindent
\textbf{Rewriting a locally uniform expression $e$ to the standard format.}
Rewriting a locally uniform expression~$e$ in some locally uniform constraint means replacing it with a new variable $z$ and adding auxiliary locally uniform constraints.
These constraints assure that the variable~$z$ will carry the desired meaning.
The result is that every expression $e$ is rewritten to the standard format.
In this phase we may still be adding constraints that are not in the standard format (contain inequalities etc.) as they will be dealt with by subsequent applications of the rewriting rules.
\begin{itemize}
    \item $e::= f \vee g \quad \Rightarrow \quad$ create a new binary variable $s$ and set $2x_e = x_f + x_g + s_e$. \\
    Then $\Delta s(e) = \Delta s(f) + \Delta s(g) + 1$, $\Delta t(e) = \Delta t(f) + \Delta t(g) + 2$.
    \item $e ::= \neg f \quad \Rightarrow \quad x_z = 1 - x_f$. \\
    Then $\Delta s(e) = \Delta s(f) + 1$, $\Delta t(e) = \Delta t(f) + 1$.
    \item $e ::= \bool_m(f)$ and $e ::= \sgn_m(f)$: We have that $-L < f < U$ for $L,U \in \N$ and $m \le \max \{L,U\}$, since $I$ is valid.
    The subscript $m$ in $\bool_m(e)$ signifies that we need to introduce coefficients (upperbounded by $m$) into the (new) system.
    Let $v_e,u_e \in \{0,1\}$ be two new binary variables such that $v_e=1$ if and only if $x_f \geq 0$, and $u_e=1$ if and only if $x_f \leq 0$:
    \begin{align*}
    1 + x_f &\leq Uv_e \leq U + x_f \\
    1 - x_f &\leq Lu_e \leq L - x_f \\
    v_e, u_e &\in \{0,1\}
    \end{align*}
    Now if $e ::= \bool_m(f)$ we let $x_z = \neg(v_e \wedge u_e)$ and if $e ::= \sgn_m(f)$ we let $x_z = v_e - u_e$. \\
    Then $\Delta s(e) = \Delta(f) + O(1)$, $\Delta t(e) = \Delta t(f) + O(1)$ and analogously for $e ::= \sgn_m(f)$.
    \item $e ::= \bool_m(f \heartsuit g) \quad \Rightarrow \quad$
    \begin{itemize}
        \item $\heartsuit$ is ``='': $x_e = \neg\bool_m(x_f - x_g)$
        \item $\heartsuit$ is ``$>$'': $x_e = \bool_m(\sgn_m(x_f - x_g) = 1)$
        \item $\heartsuit$ is ``$\geq$'': $x_e = \bool_m(x_f > x_g) \vee \bool_m(x_f = x_g)$
    \end{itemize}
    And analogously when $\heartsuit$ is ``$<$'' or ``$\leq$''. \\
    Then $\Delta s(e) = \Delta s(f) + \Delta s(g) + O(1)$ and $\Delta t(e) = \Delta t(f) + \Delta t(g) + O(1)$, since the overhead of the used operations is only $O(1)$.
\end{itemize}

\noindent
\textbf{Rewriting a locally uniform constraint $e ::= f \heartsuit g$ to standard format.}
\begin{itemize}
    \item when $\heartsuit$ is ``='': $e = f \quad \Rightarrow \qquad x_e - x_f = 0$.\\
    Then $\Delta s(e=f) = \Delta s(e) + \Delta s(f)$ and $\Delta t(e=f) = \Delta t(e) + \Delta t(f)$.
    \item when $\heartsuit$ is not ``='', intuitively we want to add a slack variable: $e \heartsuit f \quad \Rightarrow \qquad x_e - x_f + s_e = 0$ with the variable $s_e$ having an upper bound $u^s$ and a lower bound $l^s$ set as follows:
    \begin{itemize}
        \item $l^s = 0$ and $u^s = Q_e$ when $\heartsuit$ is ``$\leq$'',
        \item $l^s = 1$ and $u^s = Q_e$ when $\heartsuit$ is ``$<$'',
        \item $l^s = -Q_e$ and $u^s = 0$ when $\heartsuit$ is ``$\geq$'', and
        \item $l^s = -Q_e$ and $u^s = -1$ when $\heartsuit$ is ``$>$'';
    \end{itemize}
    where $Q_e = \max\{\|\vel\|_\infty, \|\veu\|_\infty\}\cdot \maxcoef nt$ stands for a sufficiently large number. \\
    Then, $\Delta s(e \heartsuit f) = \Delta s(e) + \Delta s(f)$ and $\Delta t(e \heartsuit f) = \Delta t(e) + \Delta t(f) + 1$.
\end{itemize}

\noindent\textbf{Rewriting a globally uniform constraint $e ::= f \heartsuit g$ to standard format.}
First, we rewrite any logical operations and $\bool_m$ and $\sgn_m$ operations in $e$ and $f$ using the same rules as above, that is, by adding auxiliary variables and locally uniform constraints.
Then, what remains is to deal with the inequalities $\heartsuit \in \{<, \leq, >, \geq\}$.
We use slack variables as before, but since we cannot add just one variable without breaking the $n$-fold format, we instead add $n$ new variables and ``disable'' all but one of them using the lower and upper bounds:\\
$e \heartsuit f \quad \Rightarrow \qquad x_e - x_f + \sum_{i=1}^n s_{e \heartsuit f}^i = 0$ with $s_{e \heartsuit f}^i$ for $i = 1,\hdots,n$ being $n$ new variables with lower and upper bounds $l^i = u^i = 0$ for $1 < i \le n$ and with
\begin{itemize}
    \item $l^1 = 0$ and $u^1 = \infty$ when $\heartsuit$ is ``$\leq$'',
    \item $l^1 = 1$ and $u^1 = \infty$ when $\heartsuit$ is ``$<$'',
    \item $l^1 = -\infty$ and $u^1 = 0$ when $\heartsuit$ is ``$\geq$'', and
    \item $l^1 = -\infty$ and $u^1 = -1$ when $\heartsuit$ is ``$>$''.
\end{itemize}
Then, $\Delta t(e \heartsuit f) = \Delta t(e) + \Delta t(f)$ if $\heartsuit$ is ``$=$'' and $\Delta t(e \heartsuit f) = \Delta t(e) + \Delta t(f) + 1$ otherwise.

It remains to compute the parameters $r',s', t'$ and $\maxcoef'$.
Let $\mathcal{L}$ be the set of locally uniform constraints and~$\mathcal{G}$ be the set of globally uniform constraints.
Then
\begin{itemize}
    \item $t' = t + \sum_{e \heartsuit f \in \mathcal{L}} \Delta t(e \heartsuit f) + \sum_{e \heartsuit f \in \mathcal{G}} \Delta t(e \heartsuit f) = t + \Oh(w)$,
    \item $s' = s + \sum_{e \heartsuit f \in \mathcal{L}} \Delta s(e \heartsuit f) = s + \Oh(w)$,
    \item $r' = r$, since we have merely added slack variables,
    \item all coefficients are bounded in absolute value by $\maxcoef'=\max\{\maxcoef,M\}$ since we have only introduced new large coefficients via the $\bool_m$ and $\sgn_m$ operations and those are upper bounded by $M$.
\end{itemize}
This concludes the proof of Theorem~\ref{thm:extendednfold}.
\end{proof}

\subsection{A demonstration of the rewriting process}
\label{sec:demonstration}
In this section we demonstrate the proof of Theorem~\ref{thm:extendednfold} on the constraints~\eqref{eqn:permut_start} and \eqref{eqn:permut_neq}.

Refer to Figure~\ref{fig:expressionRewriting}.
The first constraint is already in the standard format.
However, for the other constraints, the rewriting rules are applied.
Fix $j,k$.
The resulting standard $n$-fold IP will contain the following constraints (for brevity we omit rewriting inequalities by slack variables as this is standard):
\begin{eqnarray*}
    w &    = & x_j - x_k \\
    1 + w & \leq & mv \leq m + w \\
    1 - w & \leq & mu \leq m - w \qquad\qquad \textrm{express } z_{\neg\bool} = \neg(v \wedge u) = \neg v \vee \neg u\\
    v_\neg &   = & 1 - v \bigwedge u_\neg = 1 - u \qquad   v, u, s \in \{0,1\}\\
    2z_\vee &   = & v_\neg + u_\neg + s \\
    z_{\neg\bool} &   = & 1 - z_\vee \qquad\qquad\quad \textrm{and set } \neg\bool(x_j - x_k) = 0 \\
    z_{\neg\bool} &   = & 0 \enspace .
\end{eqnarray*}

\begin{figure}[bt]
  \includegraphics{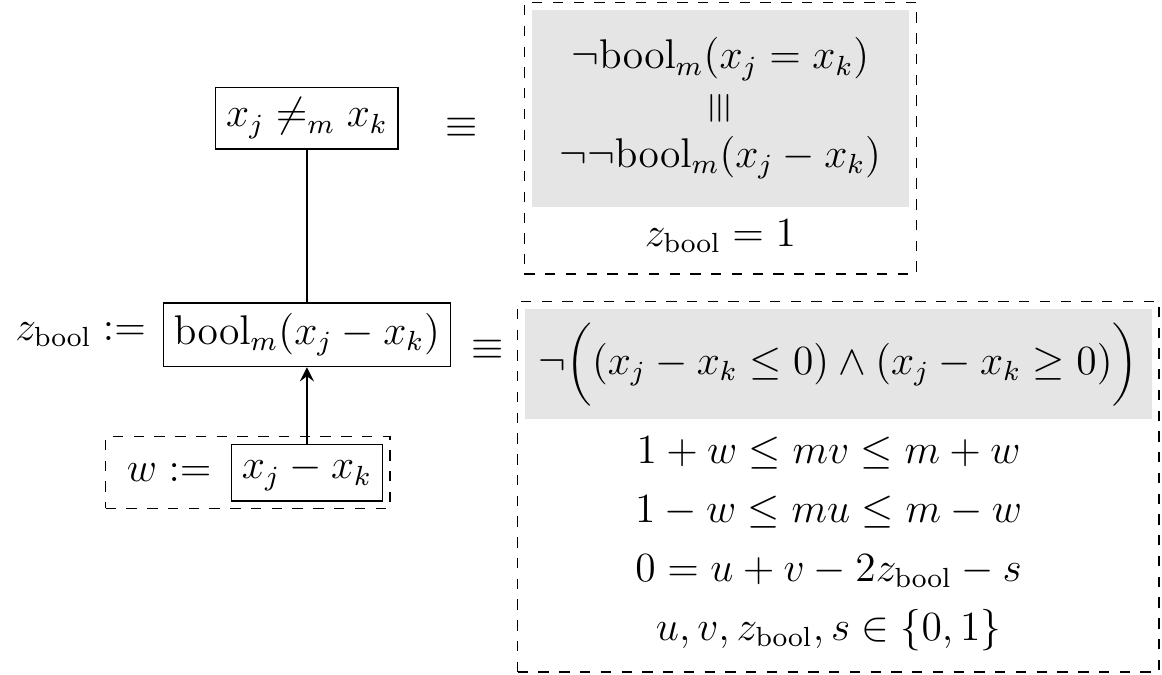}
  \caption{\label{fig:expressionRewriting}%
  An illustration of the rewriting process for \eqref{eqn:permut_neq}.
  An edge is between equivalent expressions, while an arrow is pointing towards the depending expression.
  }
\end{figure}

\paragraph{Remark.}
Naturally, we ask if the $\bool()$ operation can be implemented \textit{without} introducing a number~$\maxcoef$ depending on the lower and upper bounds, as $\maxcoef$ becomes the base of the run time in Theorem~\ref{thm:nfold}.
One can show that such dependence is necessary (proof is deferred to Section~\ref{sec:binpacking_nfold}):
\begin{lemma}
    \label{lem:nfold_bool}
    Unless $\FPT = \mathsf{W}[1]$, the $\bool()$ operation cannot be expressed in $n$-fold IP format by introducing only $f(k)$ new variables and coefficients bounded by $f(k)$, for any computable function $f$ and $k = \max\{r,s,t\}$.
    Moreover, binary $n$-fold IP is weakly NP-hard even when $r=t = 1$ and $s=0$.
\end{lemma}

\section{Single-Exponential Algorithms for Voting and Bribing}
\label{sec:applications}
We now establish a formulation of $\mathcal R$-{\sc Multi-Bribery} as an $n$-fold IP, for various rules~$\mathcal R$.
To this end, we first describe the part of the IP which is common to all such rules, in Sect.~\ref{sec:general_setup}.
Thereafter, in Sect.~\ref{sec:voting_rules} we add the parts of the formulation which depend on $\mathcal R$.

\subsection{General Setup}
\label{sec:general_setup}
Given an instance $(C,V)$ of $\mathcal R$-{\sc Multi-Bribery}, we construct an $n$-fold IP whose variables describe the situation after bribery actions (swaps, push actions, control changes) were performed.
From these variables we also derive new variables to express the cost function.
In the following we always describe the variables and constraints added per voter, and there is one brick per voter.
So fix a voter $v\in V$.

\medskip
\noindent
\textbf{Swaps.} We describe the preference order with swaps $S$ applied by variables $x_c^v, c\in C$ with the intended meaning $x_c^v = \rank(c,v)^S$.
We stress here that the ranking according to values of $x_c^v$ is the one in the altered elections.
Recall that constraints~\eqref{eqn:permut_start} and \eqref{eqn:permut_neq} enforce that $(x_1^v, \ldots, x_{|C|}^v)$ is a permutation of $C$; we add them to the program.

To express the swaps performed by $S$, for each pair of candidates $c, c' \in C$ we introduce binary variables $s_{\{c,c'\}}^v$ so that $s_{\{c,c'\}}^v=1$ if and only if $c$ and $c'$ are swapped.
We need an observation that follows from a result of Elkind et al.~\cite[Proposition 3.2]{ElkindEtAl2009b}.
\begin{observation}
    For complete preference orders $\pref,\pref'$, the admissible set $S$ of swaps such that $\pref' = \pref^S$ is uniquely given as the set of pairs $(c,c')$ for which either $c \pref c' \wedge c' \pref' c$, or $c' \pref c \wedge c \pref' c'$.
\end{observation}
Thus, we only need to set constraint~\eqref{eqn:inverted_pairs} from Sect.~\ref{sec:demonstration} with $o_c^v = \rank(c,v)$.
Further, for each pair $c,c'\in C$ of candidates we introduce a variable~$x^v_{(c,c')}$ which takes value $1$ if $c \pref_v^S c'$, and value $0$ otherwise.

\medskip
\noindent
\textbf{Push actions.}
To indicate push actions, we introduce binary variables $p_{-|C|}^v,\ldots,p_{|C|}^v$ where $p_0^v = 1$ means no change, $p_j^v = 1$ means push action $p^v = j$.
We set the lower and upper bounds to ensure that $p_j^v = 0$ for all $j \not\in \{-a^v, |C|-a^v\}$.
Finally, we introduce a variable $x_\alpha^v \in \{1,\hdots,|C|\}$ indicating $v$'s approval count after the push action:
\begin{align*}
\sum_{j=-{|C|}}^{|C|} p^v_j  &= 1, \\
x_\alpha^v  &= a^v + \sum_{j=-{|C|}}^{|C|} j p^v_j \enspace .
\end{align*}

\medskip
\noindent
\textbf{Influence bit.}
To model certain variants of the problem, such as $\mathcal R$-{\sc Manipulation}, we need an auxiliary ``influence bit''.
We introduce a binary variable $x_\iota$ taking value $1$ if a swap or a push action is performed, and value 0 otherwise:
\[
x_\iota = \bool_{|C|^2}\left(\sum_{c,c'\in C} s^v_{\{c,c'\}} + \sum_{j \neq 0} p^v_j\right) \enspace .
\]

\medskip
\noindent
\textbf{Control changes.}
We introduce two binary variables $x_a^v, x_\ell^v$ such that $x_a^v = 1,x_\ell^v = 0$ if voter $v$ is active, and $x_a^v = 0,x_\ell^v = 1$ if voter $v$ latent:
\begin{align*}
x_a^v + x_\ell^v = 1 \enspace .
\end{align*}

We will also frequently use the following variable-splitting trick:
\begin{lemma}\label{lem:splittingVars}
    Let $x$ be an integral variable with lower bound $\ell$ and upper bound $u$ and let~$z$ be a binary variable.
    It is possible to introduce a variable $x^z$, $1$ auxiliary variable, and $3$ locally uniform constraints of height at most $m=u - \ell$ such that $x^z = x$ if $z = 1$ and $x^z = 0$ if $z=0$.
\end{lemma}
\begin{proof}
Without loss of generality we assume $x$ is normalized, that is, $\ell = 0$.
If this is not the case, we replace $x$ with $(\tilde{x} + \ell)$ in all constraints and set box constraints for $\tilde{x}$ to $0 \le \tilde{x} \le u - \ell$.
Clearly, this substitution yields an equivalent integer program.

We first, add  constrains $0 \le x^z \le u z$.
Now, if $z = 0$, then $x^z = 0$ and otherwise ($z = 1$) we have $\ell \le x^z \le u$.
We now essentially repeat this trick with a negation of $z$ for a new variable $x^{\neg z}$, that is, we add constraints $0 \le x^{\neg z} \le u (1 - z)$.
Finally we add the constraint $x^z + x^{\neg z} = x$ which finishes the construction, since by the above discussion we know one of the variables $x^z, x^{\neg z}$ must be set to $0$ and thus the other must be set to the same value as $x$.

\end{proof}

\medskip
\noindent
\textbf{Objective function.}
Finally, \emph{collectively for all voters} the linear objective function is given as follows:
\begin{align*}
w(\vex, \ves, \vep) = \sum_{v\in V} \left[\left(\sum_{c,c'\in C} \sigma^v(c,c') s_{(c,c')}^v\right) + \left(\sum_{j=-|C|}^{|C|} \pi^v(j) p^v_j\right) + \iota^v x^v_\iota + \alpha^v x_a^v + \delta^v x_\ell^v\right] \enspace .
\end{align*}
Observe that the maximum coefficient of the objective function, i.e., $\|\vew\|_\infty$, is the maximum of all the cost functions $\sigma^v$, $\pi^v$, $\iota^v$, $\alpha^v$ and $\delta^v$ across all their arguments.
So far we have introduced $O(|C|^2)$ variables and imposed $O(|C|^2)$ constraints on them (per brick).
The largest coefficient introduced in a constraint is $|C|$, and we have already used the $\bool_M$ operation with $M = |C|^2$.

\subsection{Voting rules}
\label{sec:voting_rules}
Now we describe the part specific to the voting rules.
A voting rule $\mathcal R$ is incorporated in the IP in two steps.
First, optionally, new variables are derived using locally uniform constraints.
Then, globally uniform constraints are imposed.

Often we can only set up the IP knowing certain facts about how the winning condition is satisfied.
We guess those facts, construct the IP, solve it and remember the objective value.
Finally, we choose the minimum over all guesses.
%We always specify the range from which we guess.

\medskip
\noindent
\textbf{(R1) Scoring protocol} $\ves = \mathbf{(s_1, \ldots, s_{|C|})}$.
We introduce variables $\tau^v_c$ for the number of points that voter~$v$ gives candidate $c$:
\begin{align*}
\tau^v_c = \sum_{k=1}^{|C|} s_k \bool_{|C|}(x^v_c = k)
\end{align*}
Then, an ``active'' copy $\tau_c^{va}$ of each variable $\tau^v_c$ is created such that we can disregard the contribution of latent voters.
To do this we use Lemma~\ref{lem:splittingVars} with $x:= \tau^v_c$, $z:= x^a$, and $x^z:=\tau_c^{va}$ for every $c\in C$.
Then we add the following globally uniform constraints specifying that the score received by $c^\star$ is greater than the score received by any other candidate:
\begin{align*}
\sum_{v\in V} \tau^{va}_c  ~<~ \sum_{v\in V} \tau^{va}_{c^\star} \mbox{ for } c\in C\setminus\{c^\star\} \enspace .
\end{align*}

\medskip
\noindent
\textbf{(R2) Any C1 rule $\mathcal R$.}
We guess the resulting $<_M$ such that $c^\star$ is a winner with respect to~$\mathcal R$; there are $O(3^{|C|^2})$ guesses.
From $<_M$ we can infer for any pair $c,c'$ of distinct candidates, whether $v(c,c') > v(c',c)$, $v(c',c) > v(c, c')$ or $v(c,c') = v(c',c')$.
With this knowledge, we add the following constraints, where again variables with an $a$ in the superscript stand for the active parts:
\begin{align*}
\sum_{v\in V} x^{va}_{(c,c')} & > \sum_{v\in V} x^{va}_{(c',c)} & \mbox{ if } c <_M c',\\
\sum_{v\in V} x^{va}_{(c,c')} & = \sum_{v\in V} x^{va}_{(c',c)} & \mbox{ if } v(c,c') = v(c',c) \enspace .
\end{align*}

\medskip
\noindent
\textbf{(R3) Maximin rule.}
For $c^\star$ to be winner with the maximin rule means that there is a number $B \in \{0,1,\hdots,|V|\}$ such that $v_*(c^\star) = B$, while for all $c\in C\setminus\{c^\star\}$, $v_*(c) < B$.
(Recall that $v_*(c) = \min \{v(c,c') \mid c' \in C \setminus \{c\}\}$.)
This implies that for every candidate $c\not=c^\star$ there is a candidate $d(c)$ (the defeater of $c$) such that $v(c, d(c)) < B$.
All in all $B$ and $d(c)$ for every $c \in C \setminus \{c^\star\}$, that is, we guess a mapping $d\colon C \setminus \{c^\star\} \to C$; there are at most $n \cdot (|C|-1)^{|C|}$ guesses.
Then add the following constraints:
\begin{align*}
\sum_{v\in V} x^{va}_{(c^\star,c)} &\geq B &c\in C\setminus\{c^\star\}\\
\sum_{v\in V} x^{va}_{(c,d(c))} &< B &c\in C\setminus\{c^\star\} \enspace .
\end{align*}

\medskip
\noindent
\textbf{(R4) Bucklin.} To determine the control action $\gamma$, we guess the number $\left|V_a^\gamma\right|\in\{1,\hdots,|V|\}$ of active voters and set
\[\sum_{v\in V} x_a^v = \left|V_a^\gamma\right|\,.\]
Then, guess the winning round $k$ and note that the winning score will be larger than $|V_a^\gamma|/2$.
Altogether, there are $O(|C||V|)$ guesses.
Similarly to scoring protocols, we introduce variables~$\tau^v_c$ (number of points for candidate~$c$ in $k$-approval) and $\tilde{\tau}^v_c$ (number of points for candidate~$c$ in $(k-1)$-approval).
Again, we consider only the active parts $\tau^{va}_c$ of $\tau^v_c$ and $\tilde{\tau}^{va}_c$ of $\tilde{\tau}^v_c$:
\begin{align*}
\tau^{v}_c &= \bool_{|C|}\left(x_c^v < k\right) &c\in C\\
\tilde{\tau}^{v}_c &= \bool_{|C|}\left(x_c^v < k-1\right) &c\in C \enspace .
\end{align*}
Then, the winning condition is expressed as:
\begin{align*}
\sum_{v\in V} \tau^{va}_{c^\star} &> \left|{V}_a^\gamma\right|/2 \\
\sum_{v\in V} \tau^{va}_c &\leq \sum_{v\in V} \tau^{va}_{c^\star} &c\in C\setminus\{c^\star\} \\
\sum_{v\in V} \tilde{\tau}^{va}_c &\leq \left|{V}_a^\gamma\right|/2 &c\in C \enspace .
\end{align*}

\medskip
\noindent
\textbf{(R5) SP-AV.} In SP-AV, each candidate $c$ receives a point if it ranks above the approval count.
As before, we introduce variables $\tau^v_c$ for points received by a candidate $c$ and split them into active and latent (again using Lemma~\ref{lem:splittingVars}).
%Then, we guess $B \in \{1,\hdots,|V|\}$, the number of points received by $c_1$.
\begin{align*}
\tau^{v}_c &= \bool_{|C|}(x_c^v \leq x_\alpha^v) &c\in C \\
\sum_{v\in V} \tau^{va}_c &\leq \sum_{v\in V} \tau^{va}_{c^\star} &c\in C\setminus\{c^\star\} \enspace .
\end{align*}

\noindent
\textbf{(R6) Fallback.} In the fallback rule, the non-approved candidates are discarded, the Bucklin rule is applied
and if it fails to select a winner, the SP-AV rule is applied.
We guess the Bucklin winning round $k$ or if SP-AV is used and the number $\left|{V}_a^\gamma\right|$ of active voters; there are $O(|V|\cdot|C|)$ guesses.
(SP-AV is used exactly when the winning score is less than $\left|{V}_a^\gamma\right|/2$.)
If Bucklin is used, we need a slight modification to take push actions into account.
Instead of $\tau^{va}_c = \bool_{|C|}(x_c < k)$ we have
\[
\tau^{v}_c = \bool_{|C|}(x^v_c < k) \wedge \bool_{|C|}(k \leq x^v_\alpha) \,;
\]
similarly for~$\tilde{\tau}^{va}_c$.

\medskip
For each of the rules (R1)-(R6), as argued, we have constructed an extended $n$-fold IP with  $O(|C|^2)$ variables and locally uniform constraints per brick, and $O(|C|^2)$ globally uniform constraints.
The largest coefficient is $|C|$, the height is $\Oh(|C|^2)$ and the extended width is also $\Oh(|C|^2)$.
Thus, by Theorem~\ref{thm:extendednfold} we can compute an $n$-fold matrix with parameters $r=s=t=\maxcoef=O(|C|^2)$.
Proposition~\ref{thm:nfold} is then used to solve this $n$-fold IP in time $2^{O(|C|^6 \log|C|)}n^3 \langle \vew \rangle$.
Also, $O(3^{|C|^2})$ guesses suffice for each rule except Maximin, Bucklin, and Fallback, where $O(|C|^2|V|)$ guesses suffice.

\medskip
An exception to this run time is the Kemeny rule:

\noindent\textbf{(R7) Kemeny.}
For $c^\star$ to be a Kemeny winner, there has to be a ranking $\pref_{R^*}$ of the candidates that ranks~$c^\star$ first and maximizes the total agreement with voters
\[
\sum_{v\in V} \bigl|\left\{ \left(c,c'\right) \in C \times C \mid ((c \pref_{R^*} c') \Leftrightarrow (c \pref_v c')) \right\}\bigr|
\]
among all rankings.
In other words, the number of swaps sufficient to transform all $\pref_i$ into $\pref_{R^*}$ is smaller than the number of swaps needed to transform all $\pref_v$ into any other $\pref_{R'}$ where $c^\star$ is not first.

We guess the ranking $\pref_{R^*}$ over all rankings of $C$; then we introduce variables $x^v_R$ for $R \in \{R^*\} \cup \{R' \mid c^\star \textrm{ is not first in } R'\}$ so that $x^v_R$ is the number of swaps needed to transform $\pref^S_v$ into $\pref_R$ (recall that $S$ is the bribery described by the variables $x_c^v$).
We again split variables $x^v_R$ and in the score comparison consider their active parts $x_R^{av}$ only.
Then, we introduce the necessary constraints:
\begin{align*}
x_R^v & =~\sum_{c, c'\in C, c \neq c'} \bool_2\left(\sgn_{|C|}(x_c^v - x_{c'}^v) = \sgn_{|C|}(\rank(c',R) - \rank(c,R)\right) & \mbox{ for all } R, \\
\sum_{v\in V} x^{av}_R & >~\sum_{v\in V} x^{av}_{R^*}, & R \neq R^* \enspace .
\end{align*}
This solves Kemeny-{\sc Multi-Bribery} in time $|C|^{O(|C|!^6)}n^3$, and completes proving Theorem~\ref{thm:main-mixedplusbribery}.

\section{Lower Bounds and Hardness for \texorpdfstring{$n$-fold}{n-fold} IPs}
\label{sec:lowerbounds}
Here we provide the proof of Lemma~\ref{lem:nfold_bool} which shows that the $\bool()$ operation cannot be implemented in $n$-fold IPs without introducing large numbers into the system and that solving $n$-fold IPs becomes $\mathsf{W}[1]$-hard when parameterized only by $(r,s,t)$.

\subsection{\texorpdfstring{$\bool()$}{bool()} inexpressibility: \texorpdfstring{\textsc{Unary Bin Packing}}{Unary Bin Packing}}
\label{sec:binpacking_nfold}
The \textsc{Unary Bin Packing} problem takes as input $n$ items of integer sizes $o_1, \ldots, o_n$ as well as two integers $k,B$, and asks if the items can be packed into $k$ bins each of which has capacity $B$.
Here by packing we mean an assignment of items to bins $\sigma\colon \left\{ 1, \ldots, n \right\} \to \left\{ 1, \ldots, k \right\}$ such that the packing is \emph{admissible}, that is, $\sum_{i \in \sigma^{-1}(j)} o_i \le B$ for every $j \in \left\{ 1, \ldots, k \right\}$.

\defparproblem{\textsc{Unary Bin Packing}}
{$k$}
{Positive integers $k,B$ encoded in unary and item sizes $o_1, \ldots, o_n$ for every $i \in \left\{ 1, \ldots, n \right\}$.}
{Find an admissible packing of the items to $k$ bins.}

Jansen et al.~\cite{JansenEtAl2013} prove that this problem is $\mathsf{W}[1]$-hard parameterized by~$k$.

\begin{lemma}\label{lem:bool_is_hard}
Unless $\FPT = \mathsf{W}[1]$, the $\bool()$ operation cannot be expressed in $n$-fold IP format by introducing only $f(k)$ new variables and numbers bounded by $f(k)$, for any computable function $f$ and $k = \max\{r,s,t\}$.
\end{lemma}
\begin{proof}
Given an instance $(o_1,\hdots,o_n,k,B)$ of {\sc Unary Bin Packing}, we create an $n$-fold IP as follows.
We create a brick for each item $o_i$, and introduce $k$ variables $x^i_1, \ldots, x^i_k$ for $i = 1,\hdots,n$ and the following locally uniform constraints:
\begin{align}
\sum_{j=1}^k x^i_j = o_i, \tag{AssignItem}\label{eq:assign_item} \\
\sum_{j=1}^k \bool_n(x^i_j) = 1 \tag{OneBin}\label{eq:one_bin} \enspace .
\end{align}

Clearly \eqref{eq:assign_item} and \eqref{eq:one_bin} together forces $x^i_j = o_i$ for exactly one $j$, which we then denote $j(i)$, and $x^i_{j'} = 0$ for all $j' \neq j(i)$.
Intuitively, this means that an item $i$ belongs to the bin $j(i)$.
Using globally uniform constraints we enforce that no bin overflows:
\begin{equation*}
\sum_i x_j^i \leq B, \qquad j = 1,\hdots,k \enspace .
\end{equation*}
It is clear that this IP is feasible if and only if $(o_1,\hdots,o_n,k,B)$ is a ``yes'' instance.

Now suppose it is possible to express the $\bool_n(\cdot)$ using $f(k)$ additional local variables and $f(k)$ locally uniform constraints all of which use numbers bounded by $f(k)$ in absolute value.
Replacing the condition \eqref{eq:one_bin} with this expression and invoking Theorem~\ref{thm:extendednfold} on the thus altered IP model we obtain an $n$-fold IP with $r,s,t,A = f(k)$.
Finally, this would yield an $f'(k) n^{\Oh(1)}B^{\Oh(1)}$-time algorithm for the \textsc{Unary Bin Packing} problem, where $f'$ is a computable function independent of~$n$.
Consequently, \textsf{FPT}$ = $\textsf{W}[1].
\end{proof}

\subsection{Largest Coefficient Matters: Subset Sum}
The \textsc{Subset Sum} problem is a well known (weakly) \textsf{NP}-hard problem and is defined as follows.
Given $n$ positive integers $w_1, \ldots, w_n$ and the target value $T$ the task is to find a set $I \subseteq \left\{ 1, \ldots, n \right\}$ such that $\sum_{i \in I} w_i = T$.

\begin{lemma}\label{lem:nfold_subsetsum}
Binary $n$-fold IP is weakly NP-hard even when $r=t = 1$ and $s=0$.
\end{lemma}
\begin{proof}
We formulate the \textsc{Subset Sum} problem straightforwardly as an $n$-fold IP with $n$ bricks a exactly one global condition:
\begin{equation}
\sum_{i \in \left\{ 1, \ldots, n \right\}} w_i x^i_1  = T \,,
\end{equation}
where $x^i_1 \in \{ 0,1 \}$ are binary variables.
Observe that the parameters of the above $n$-fold IP are $r = t = 1$ and $s = 0$.
\iffalse
Suppose now we have an algorithm that solves $n$-fold IP in the claimed time $f(r,s,t) \cdot (nL)^{O(1)}$.
Since all of $r,s,t$ are bounded by a constant, it follows that $f(r,s,t)$ is a constant.
Let $c = f(1,0,1)$.
Using such a hypothetical algorithm we can decide \textsc{Subset Sum} in time $c \cdot (nL)^{O(1)}$, where $L$ is the length of binary encoding of the vector $(w_1, \ldots, w_n, T)$.
Consequently, $\mathsf{P} = \mathsf{NP}$.
\fi
\end{proof}

\begin{proof}[Proof of Lemma~\ref{lem:nfold_bool}]
The lemma is a direct consequence of Lemma~\ref{lem:bool_is_hard} and Lemma~\ref{lem:nfold_subsetsum}.
\end{proof}

\section{Conclusions and open problems}
\label{sec:conclusions}
We introduced a general voting and bribing problem, $\mathcal R$-{\sc Multi-Bribery}, which allows for swaps, push actions, and control changes.
For several classical voting rules $\mathcal R$, we provided formulations of $\mathcal R$-{\sc Multi-Bribery} in terms of $n$-fold integer programs; those formulations lead
\begin{itemize}
  \item to the first fixed-parameter algorithms for some of those problems and
  \item to the first single-exponential algorithms for others.
\end{itemize}
Our approach is also natural in handling situations where each voter has different pricing functions, which was previously not possible in most cases.
In particular, we provide the first fixed-parameter algorithm for $\mathcal R$-{\sc Swap Bribery} with arbitrary cost functions, for many natural voting rules $\mathcal R$.

While our result covers many classical and well-studied voting rules, we are convinced that many other rules are covered by our framework.
What would be highly desirable though is to obtain sufficient conditions on easy-to-check properties of rules $\mathcal R$ for which $\mathcal R$-{\sc Multi-Bribery} is fixed-parameter tractable parameterized by the number of candidates, or for which it admits a single-exponential time algorithm.

It would also be desirable to complement our algorithmic upper bounds with matching lower bounds based on the Exponential Time Hypothesis, by either improving the run times of the provided algorithms or providing appropriate hardness results.

\paragraph{Acknowledgements.}
We are grateful to the anonymous reviewers of a preliminary version of our paper which appeared in the proceedings of STACS 2017~\cite{KnopEtAl2017}, for their helpful comments which led to a considerably improved presentation of our results here.
We also express our gratitude towards Ildik{\'o} Schlotter for many helpful remarks.

% Bibliography
%\bibliographystyle{ACM-Reference-Format}
\bibliographystyle{plain}
\bibliography{sched}

\clearpage
\appendix
%%%%%%%%%%%%%%
\section{Problem Definitions}
\label{sec:problemdefinitions}
We provide the formal definitions of the voting and bribing problems covered by our result for \textsc{$\mathcal R$-Multi-Bribery}.
In these definitions, for each problem $\mathcal R$-{\sc Problem}, by ``such that $c^\star$  wins the election'' we mean that $c^\star$ wins the election under voting rule $\mathcal R$.

The $\mathcal R$-{\sc \$Bribery} problem was introduced by~Faliszewski et al.~\cite{FaliszewskiHH06}, and further studied by Bredereck et al.~\cite{BredereckEtAl2015}:

\defparproblem{$\mathcal R$-\$Bribery}{$|C|$}{A complete election $(C,V)$, a designated candidate $c^\star\in C$, and a price vector $(p_v)_{v\in V}$.}{Select a subset $S \subseteq V$ minimizing $\sum_{v \in S} p_v$ such that changing their preference orders $\pref_v,v\in S$ arbitrarily to preference orders $\pref_v'$ makes $c^\star$ win the election $\left(C,(\pref_v')_{v\in S}\cup(\pref_v)_{v\in V\setminus S} \right)$.}

The $\mathcal R$-{\sc Manipulation} problem was introduced by Faliszewski et al.~\cite{FaliszewskiEtAl2008b}, who studied it for the Copeland$^{\alpha}$ rule:

\defparproblem{$\mathcal R$-Manipulation}{$|C|$}{An election $(C,V)$, a designated candidate $c^\star\in C$, a set $M \subseteq V$ of manipulators whose preference lists are empty (all candidates tied) and complete preferences for all $V\setminus M$.\\}{Determine complete preference orders $\pref'_v$ for all voters $v \in M$ such that $c^\star$  wins the election $(C,(\pref'_v)_{v\in M}\cup(\pref_v)_{v\in V\setminus M})$.}

The $\mathcal R$-{\sc Swap Bribery} problem was introduced by Elkind et al.~\cite{ElkindEtAl2009b}:

\defparproblem{$\mathcal R$-Swap Bribery}{$|C|$}{A complete election $(C,V)$, a designated candidate $c^\star\in C$, and swap cost functions $\sigma^v$ for $v\in V$.}{Find a set $S$ of admissible swaps with minimum cost such that $c^\star$ wins the election $(C, V)^S$.}

\defparproblem{$\mathcal R$-Shift Bribery}{$|C|$}{An election $(C,V)$, a designated candidate $c^\star\in C$, and swap cost functions $\sigma^v$ for $v\in V$.}{Find a set $S$ of admissable shifts of smallest cost such that $c^\star$ wins the election $(C, V)^S$.}

\defparproblem{$\mathcal R$-Support Bribery}{$|C|$}{An election $(C,V)$ with approval counts $a^v$ for $v\in V$, a designated candidate $c^\star\in C$, and push action cost functions $\pi^v$ for $v\in V$.\\}{Find a set $P$ of push actions with minimum cost such that $c^\star$ wins the election $(C, V)^P$.}

\defparproblem{$\mathcal R$-Mixed Bribery}{$|C|$}{An election $(C,V)$ with approval counts $a^v$ for $v\in V$, a designated candidate $c^\star\in C$, swap cost functions $\sigma^v$ for $v\in V$ and push action cost functions $\pi^v$ for $v\in V$.\\}{Find a set $S$ of admissible swaps and a set $P$ of push actions with minimum cost such that $c^\star$ wins the election $((C,V)^S)^P$.}

\defparproblem{$\mathcal R$-Extension Bribery}{$|C|$}{An election $(C,V)$ with approval counts $a^v$ for $v\in V$ and $a^v$-top-truncated preference orders, and push action cost functions $\pi_v$ for $v\in V$.\\}{Extend the approved part of each voter by $k_v \in \N$ previously disapproved candidates such that $c^\star$ becomes the winner, minimizing $\sum_{v\in V} c_v(k_v)$.}

\defparproblem{$\mathcal R$-Possible Winner}{$|C|$}{An election $(C,V)$ and a designated candidate $c^\star\in C$.}{Decide if the partial orders $\prec_v,v\in V$ can be extended to linear orders $\prec_v'$ such that $c^\star$ wins the election $(C,\{\prec_v\}_{v\in V})$.}

\defparproblem{$\mathcal R$-CCAV/CCDV}{$|C|$}{An election $(C,V=V_a \uplus V_\ell)$, a designated candidate $c^\star\in C$, activation costs $\alpha^v$ for $v\in V_{\ell}$ and deactivation costs $\delta^v$ for $v \in V_a$.\\}{Find a set of voter (de)activations with minimum cost so that $c^\star$ wins the perturbed election.}

\defparproblem{Dodgson Score}{$|C|$}{An election $(C,V)$ and a designated candidate $c^\star\in C$.}{Find a smallest set $S$ of admissable swaps such that $c^\star$ becomes the Condorcet winner of the election $(C,V)^S$.}

\defparproblem{Young Score}{$|C|$}{An election $(C,V)$ and a designated candidate $c^\star\in C$.}{Find a smallest set $S\subseteq V$ of voters that need to be removed such that $c^\star$ becomes the Condorcet winner of the election $(C,V \setminus S)$.}

\end{document}